\documentclass[conference,letter]{IEEEtran}

\IEEEoverridecommandlockouts

\usepackage{color}
\usepackage[dvipsnames]{xcolor}
\usepackage{balance}
\usepackage{soul}
\usepackage{amsthm}
\usepackage{mathtools}

\usepackage{bbm}
\usepackage{amssymb}
\usepackage{tikz}
\usepackage{enumitem}
\usepackage[algo2e,ruled,vlined,linesnumbered]{algorithm2e}
\SetInd{0em}{0.5em}
\SetKwComment{Comment}{/\!\!/}{}
\usepackage{booktabs}
\usepackage{cite}
\newcommand\blfootnote[1]{%
  \begingroup
  \renewcommand\thefootnote{}\footnote{#1}%
  \addtocounter{footnote}{-1}%
  \endgroup
}
\newtheorem{Theorem}{Theorem}
\newtheorem{Proposition}[Theorem]{Proposition}
\newtheorem{Lemma}[Theorem]{Lemma}

\theoremstyle{definition}

\DeclareMathOperator*{\argmin}{arg\,min}

\usepackage{caption}
\begin{document}

\title{ On Optimal Multi-user Beam Alignment \\     in Millimeter Wave Wireless Systems}

\author{Abbas Khalili$^\dagger$, Shahram Shahsavari$^\ddagger$,Mohammad A. (Amir) Khojastepour$^\diamond$, Elza Erkip$^\dagger$\\
 $^\dagger$NYU Tandon School of Engineering,  $^\ddagger$University of Waterloo, $^\diamond$NEC Laboratories America, Inc. \\
Emails: $^\dagger$\{ako274, elza\}@nyu.edu, $^\ddagger$shahram.shahsavari@uwaterloo.ca,
$^\diamond$amir@nec-labs.com }

\maketitle
\begin{abstract}
Directional transmission patterns (a.k.a. \textit{narrow beams}) are the key to wireless communications in millimeter wave (mmWave) frequency bands which suffer from high path loss and severe shadowing. In addition, the propagation channel in mmWave frequencies incorporates only a few number of spatial clusters requiring a procedure to align the corresponding narrow beams with the angle of departure (AoD) of the channel clusters. The objective of this procedure, called \textit{beam alignment} (BA) is to increase the beamforming gain for subsequent data communication.
Several prior studies consider optimizing BA procedure to achieve various objectives such as reducing the BA overhead, increasing throughput, and reducing power consumption. While these studies mostly provide optimized BA schemes for scenarios with a single active user, there are often multiple active users in practical networks. Consequently, it is more efficient in terms of BA overhead and delay to design multi-user BA schemes which can perform beam management for multiple users collectively. This paper considers a class of multi-user BA schemes where the base station performs a one shot scan of the angular domain to simultaneously localize multiple users.
The objective is to minimize the average of expected width of remaining uncertainty regions (UR) on the AoDs after receiving users' feedbacks. Fundamental bounds on the optimal performance are analyzed using information theoretic tools. Furthermore, a BA optimization problem is formulated and a practical BA scheme, which provides significant gains compared to the \textit{beam sweeping} used in 5G standard, is proposed. \blfootnote{This work is supported by
National Science Foundation grants EARS-1547332, 
SpecEES-1824434, and NYU WIRELESS Industrial Affiliates.}
\end{abstract}

\section{Introduction}
Millimeter wave (mmWave) frequency bands provide large bandwidth which can be used to achieve multi-Gbps throughputs in next generation wireless networks \cite{mmWave-survey-nyu}. High path loss and intense shadowing are among the major obstacles to increase data rate in such high frequency bands \cite{rappaport2011state}. To overcome these effects, several beamforming (BF) techniques have been proposed to concentrate the transmitted energy toward the point of interest using directional transmission patterns, i.e., narrow beams \cite{kutty2016beamforming}. On the other hand, experimental studies reveal that mmWave channel incorporates a few spatial clusters due to wave propagation properties in mmWave frequencies \cite{akdeniz2014millimeter}. Therefore, it is crucial to devise beam alignment (BA) techniques (also known as \textit{beam search} and \textit{beam training}) to align the narrow transmission beams with the direction of channel, i.e., the angle of departure (AoD) associated with channel clusters \cite{giordani2018tutorial}.

A slight misalignment between the beam direction and the AoD of the channel can diminish the BF gain required for high data rates when using narrow beams \cite{Lee2018,nitsche2015steering,Shah1906:Robust}. As a result, a wide range of BA techniques have been proposed to find the best beam direction \cite{barati2016initial,giordani2016comparative,desai2014initial,hussain2017throughput,shah-pimrc,Shah1906:Robust,khosravi2019efficient,chiu2019active,shabara2018linear,klautau20185g,Song2019}. These techniques can be used in the initial access procedure where the users try to connect to the base station (BS) and get access to network resources \cite{barati2016initial,giordani2016comparative}, or after initial access for beam tracking where the goal is to adjust the beam directions to compensate for the changes in angle of arrival (AoA) or AoD caused by mobility or the propagation environment \cite{Shah1906:Robust,michelusi2018optimal}. \textit{Exhaustive search} (ES) algorithms (also known as \textit{beam sweeping}) scan various directions  using a set of beams (usually from a codebook), and choose the beam with the highest received power \cite{barati2016initial,giordani2016comparative}. To reduce the overhead incurred by ES, \textit{hierarchical search} (HS) algorithms are proposed which search wider sectors first using coarser beams, and then refine the search within the best sector   \cite{desai2014initial,hussain2017throughput}.

BA schemes can be categorized into two classes, namely \textit{non-interactive} BA (NI-BA) and \textit{interactive} BA (I-BA) methods. In NI-BA, the scanning phase occurs first by the transmitter and the receiver's feedback is sent after performing all the scans. ES is an example for NI-BA. In I-BA, on the other hand, the receiver sends a number of feedback messages during the scanning phase and the transmitter uses that information in the subsequent scans. An example of I-BA is HS.

BA procedure (i.e. the angular span and location of the probed sectors) can be optimized to achieve various objectives such as minimizing power consumption \cite{michelusi-efficient,hussain2017energy} or maximizing average throughput \cite{shah-pimrc,hussain2017throughput}.
While I-BA can potentially lead to a higher performance compared to NI-BA \cite{michelusi2018optimal}, most of the optimized I-BA algorithms proposed in the literature consider the scenario with a single user  \cite{hussain2017throughput,shah-pimrc} or two users \cite{hassan2018multi} and it is not straightforward to generalize them for multi-user settings. Considering there are often multiple users in practical networks, it is more efficient in terms of BA overhead and delay to devise optimized BA strategies which adjust the beam for multiple users simultaneously. Moreover, optimizing NI-BA is more tractable in multi-user scenarios as the scanning phase is independent of the users' feedback and can be performed once for all users.

In this paper, we investigate the problem of optimal single-cell multi-user NI-BA where the users have omnidirectional transmission and reception patterns during the BA phase and the BS can steer the beams using a massive antenna array in mmWave bands. We formulate an optimization problem to find the NI-BA procedure that minimizes the weighted average expected widths of remaining uncertainty regions (UR) on the AoDs where the weights are chosen based on user priorities. We provide information theoretic upper and lower-bounds on the optimal performance and formulate optimization problems both for unconstrained beams and more practical scenario of contiguous beams. We show that our proposed schemes, which are given by the solution to the optimization problems, provide significant gains compared to the ES scheme used in 5G standard. 

\vspace*{-0.2cm}
\section{System Model and Preliminaries} 
\label{sec:sys}
We consider the downlink of a single-cell mmWave scenario including
a BS and $N$ users.
We consider the case where the propagation channel between the BS and the users consists of one path (spatial cluster). This assumption has been adopted in several prior studies  \cite{hussain2017throughput,michelusi-efficient}, and is supported by experimental studies \cite{akdeniz2014millimeter}.
We assume that the channels are stationary in the time interval of interest. Let $\Psi_j, j\in[N]$ \footnote{We denote by $[N]$ the set of integers from 1 to $N$.} denote the random AoD corresponding to the channel from the BS to the $j^{th}$ user. We consider an arbitrary probability distribution function (PDF) $f_{\Psi_j}(\cdot)$ for $\Psi_j$. This distribution reflects the prior knowledge about the AoD, which for example could correspond to the previously localized AoD in beam tracking applications. We assume that the support of the PDF $f_{\Psi_j}(\cdot)$ is $(0, 2\pi]$ for all $j\in[N]$.

We assume that the BS has a massive antenna array as envisioned for mmWave communications \cite{mmWave-survey-nyu} while the user has an omni-directional transmission and reception pattern for the BA phase similar to \cite{hussain2017throughput, hassan2018multi}. Furthermore, as in \cite{hussain2017throughput,hussain2019energy}, we consider a single RF chain along with analog BF at the BS due to practical considerations such as power consumption. To model the directionality of the BS transmission due to BF, we adopt a \textit{sectored antenna} model from \cite{ramanathan2001performance}, characterized by two parameters: a constant main-lobe gain and the angular coverage region (ACR) which is the union of the angular interval(s) covered by the main-lobe. This model has been adopted in several other studies to model the BF gain of antenna arrays (e.g. \cite{bai2015coverage,fund2017,fund2016}) or the radiation pattern of the directional antennas (e.g. \cite{shah2018,shah2017}). We neglect the effect of the side-lobes. While this ideal model is considered for theoretical tractability, modifications such as the ones presented at \cite{shah-pimrc} may be applied to generalize the antenna model for practical scenarios where the beam pattern roll-off is not sharp. 
\begin{figure}[t]
\centering
\includegraphics[width=0.7\linewidth, draft=false]{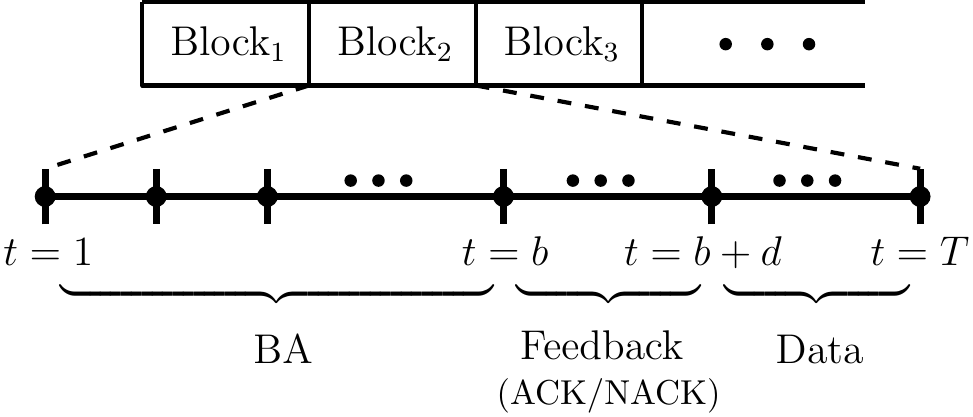}
\caption{ Time slotted system.}
\label{fig:sys_model}
\vspace*{-0.7cm}
\end{figure}

We  consider  a  time  division  duplex  (TDD)  protocol in which time is divided into blocks of $T$ time slots as depicted in Fig. \ref{fig:sys_model}. At each block, the first $b+d$ slots are used for NI-BA where the BS uses the first $b$ slots to scan the angular space and the next $d$ slots are used to collect the users' feedback (including the processing delays). The final $T-b-d$ slots are used for data communications. In this paper, we focus on the BA phase whose main objective is to localize the users' AoDs as accurately as possible such that narrower beams (which are translated to higher BF gains) can be used to serve the users during data communication phase. At each time slot $i\in[b]$ during the BA phase, the BS transmits a probing packet using a beam with ACR $\Phi_i$ to scan that angular region. While the users are silent during the first $b$ time slots, in the next $d$ slots, each user sends a feedback message to the BS (e.g. through low-frequency control channels \cite{mmWave-survey-nyu}) including the indices of the beams whose corresponding probing packets are received by that user. Therefore, this feedback message can determine if the AoD of the user belongs to the ACR of each beam, considered as Acknowledgment (ACK), or not, considered as negative ACK (NACK). In this paper, we assume that the packets used for scanning during BA are received without noise at the users. The more general case which includes the effect of noise is left for future research. Similarly, we assume that the feedback received at the BS is without any error \cite{hussain2017throughput,shah-pimrc,Shah1906:Robust}.
Based on the feedback, the BS infers an UR for the AoD of the channel of each user defined by $B(\Psi_j)$ for $j^{{\rm th}}$ user. During the data communication phase, a beam with an ACR covering the UR is used. For a probing beam with ACR $\Phi_i$, let $\Theta_{i}(\Psi_j) = \Phi_i$ if $\Psi_j \in \Phi_i$ (ACK), and $\Theta_{i}(\Psi_j) = (0,2\pi] - \Phi_i$ otherwise (NACK). Then, we have
\begin{align}
\label{eq:phis}
B(\Psi_j) = \cap_{i = 1}^b \Theta_{i}(\Psi_j), j\in[N].
\end{align}

For any BA strategy with a given set of scanning beams $\{\Phi_i\}_{i=1}^b$, let $\{u_k\}_{k=1}^{M}$ denote the set of all possible values for $B(\Psi_j)$, for all $j\in[N]$. The set $\{u_k\}_{k=1}^{M}$ does not depend on the users and is only a property of $\{\Phi_i\}_{i=1}^b$. Note that $M\leq 2^b$ and its exact value depends on the set of scanning beams. It can be shown that $u_k$'s partition the interval $(0,2\pi]$, i.e. $u_k\cap u_{k'}=\varnothing, \forall k, k'$ and $\cup_{k=1}^M u_k =[0, 2\pi)$. Given a set of weights $\{w_j\geq 0\}_{j=1}^N, \sum_{j=1}^N w_j = 1$ for users and partition $\{u_k\}_{k=1}^M$, we define the weighted average of expected width of users' URs as
\vspace*{-0.2cm}
\begin{align}
\label{eq:obj}
\overline{U}\left(\{\Phi_i\}_{i=1}^b\right)=&\sum_{j=1}^N w_j\mathbb{E}[|B(\Psi_j)|],\!\! \quad \!\!
\end{align}
\vspace*{-0.4cm}
\begin{align}
\label{eq:part}
\text{where,} \qquad &\mathbb{E}[|B(\Psi_j)|] = \sum_{k=1}^{M} |u_k|\int_{\psi\in u_k}\!\! f_{\Psi_j}(\psi) d\psi,\!\qquad  \\
&B(\psi) = 
u_k \quad \mathrm{for} \quad \psi\in u_k. \label{eq:B-def}
\end{align}
Note that $|u_k|$ is the Lebesgue measure of $u_k$, which would correspond to the total width of the intervals in case when $u_k$ is the union of a finite number of intervals. The dependence of $\overline{U}$ on $\{\Phi_i\}_{i=1}^b$ is in the expectation which is a function of partition $\{u_k\}_{k=1}^M$ created by $\{\Phi_i\}_{i=1}^b$. The weights here could reflect priorities of the users, with $w_j=1/N, \forall j\in[N]$ representing the equal priority case. The objective is to design the scanning beams $\{\Phi_i\}_{i=1}^b$ for a given $b$ to minimize the weighted average of expected width of users' URs defined in \eqref{eq:obj}. In other words, the goal is to solve the following optimization problem:
\begin{align}
\label{eq:optimization}
    \{\Phi_i^*\}_{i=1}^b = \argmin_{\{\Phi_i\}_{i=1}^b} \overline{U}\left(\{\Phi_i\}_{i=1}^b\right).
\end{align}
Since the beams used for data communications cover the corresponding URs, the objective of the optimization is equivalent to minimizing the average of the beamwidth used for data communications with users. We define $\overline{U}^*$ as the optimal objective value in optimization \eqref{eq:optimization}.

\section{Performance Analysis and BA Schemes} 
\label{sec:res}

In this section, we provide upper and lower-bounds and achievability schemes on $\overline{U}^*$ and associated BA schemes.
To this end, we first show that the multi-user BA problem can be translated to the problem of minimizing the $\overline{U}^*$ when there is only one user whose AoD PDF is equal to the weighted average of the PDFs of the users. Then, we characterize the properties of the optimal partition $\{u_k\}_{k=1}^M$ and provide bounds on the performance of the system along with a method for designing the optimal beams. 
\begin{Lemma}
\label{lem:eq}
The multi-user NI-BA problem with the objective function \eqref{eq:obj} is equivalent to a single user NI-BA problem with the objective of minimizing the expected width of the partition $\{u_k\}_{k=1}^M$ and the following prior on the AoD of that user:
\vspace*{-0.1cm}
\begin{align}
\label{eq:equivalence}
     f_{\Psi}(\psi) =  \sum_{j=1}^N w_j f_{\Psi_{j}}(\psi), \quad \psi \in (0,2\pi].
\end{align}
\end{Lemma}
\begin{proof}
Given any partition  $\{u_k\}_{k=1}^M$ of the interval $(0,2\pi]$ and weights $\{w_j\}_{j=1}^N$. Following \eqref{eq:part},
\begin{align}
    \sum_{j=1}^N &w_j\mathbb{E}[|B(\Psi_j)|] = \sum_{j=1}^N \left [w_j\sum_{k=1}^M |u_k| \int_{x\in u_k} f_{\Psi_{j}}(x) dx \right]\\
    =&\sum_{k=1}^M |u_k| \int_{z\in u_k} \left[ \sum_{j=1}^N w_j f_{\Psi_{j}}(z)\right] dz\! = \mathbb{E}[|B(\Psi)|],
\end{align}
where the expectation is over the PDF in \eqref{eq:equivalence}.
\end{proof}

Using Lemma \ref{lem:eq}, for the rest of the paper, we only consider the BA problem for one user with the PDF in \eqref{eq:equivalence}.
It is straightforward to show that the optimal scanning beams in the optimization problem \eqref{eq:optimization} are the ones that can generate an optimal partition $\{u_k^*\}_{k=1}^{M}$ such that 
\vspace*{-0.1cm}
\begin{align}
\begin{aligned}
\label{eq:optimization2}
\{u_k^*\}_{k=1}^{M} = \argmin_{u_k, k\in [M]} \quad  &\sum_{k=1}^{M} |u_k|\mathbb{P}(\Psi \in u_k)\\
\textrm{s.t.} \quad & |u_k| \geq 0. 
\end{aligned}
\end{align}

To bound $\overline{U}^*$, we note that the partition $\{u_k\}_{k=1}^M$ created by the beams used for scanning can be viewed as a set of quantization regions and the optimization in \eqref{eq:optimization2} is equivalent to designing a quantizer that minimizes the expected width (i.e. Lebesgue measure) of the partition $\{u_k\}_{k=1}^M$.

\begin{Proposition}
\label{prop:quant}
Consider a scalar quantizer with quantization regions $\{u_k\}_{k=1}^M$ that partition an interval $[a,b]$ and a real valued random variable $X$, defined on the interval $[a,b]$. Then
\vspace*{-0.3cm}
\begin{align}
\label{eq:quant}
   2^{h(X)-\log{M}} \leq \mathbb{E}[|Q(X)|],
\end{align}
where $Q(x) = u_k$ for $x \in u_k$, $|u|$ refers to the Lebesgue measure of the set $u$, and $h(X)$ is the differential entropy of $X$. This bound is tight when $X$ is a uniform random variable.
\end{Proposition}
\begin{proof}
\vspace*{-0.1cm}

Define the discrete random variable $ \hat{X}=k$ if $ X\in u_k$. Then $\hat{X}$ has the probability distribution
\begin{align}
    p_i = \mathbb{P} (\hat{X} = k ) = \mathbb{P}\left (X \in u_k\right) \quad \mathrm{for} \quad k\in[M].
\end{align}
We set  
\vspace*{-0.3cm}
\begin{align}
\bar{f}_X(x) = \frac{\mathbb{P} \left (X\in u_{k} \right)}{|Q(x)|}\ \mathrm{for}\quad x\in u_k.
\end{align}
\vspace*{-0.5cm}
\begin{align}
&H(\hat{X}) =\sum_{k=1}^M \!-p_k \log{p_k}  \\
&= \sum_{k=1}^M \int_{u_k} -f_X(x)\log{\left(\bar{f}_X(x)|Q(x)|\right)} dx \\
&=\int_{a}^b -f_X(x)\log{\left(\bar{f}_X(x)|Q(x)|\right)} dx \\
&= \int_{a}^b\!-f_X(x)\log{f_X(x)} dx +\int_{a}^b -f_X(x)\log{|Q(x)|} dx \notag\\
&\quad+ \int_{a}^b\! -f_X(x)\log{\frac{\bar{f}_X(x)}{f_X(x)}} dx \\
&=  h(X)- \mathbb{E}[\log{|Q(X)|}]+ D_{KL}\left(f_X(x)\| \bar{f}_X(x) \right).
\end{align}
Using $H(\hat{X})\leq \log{M}$, $ D_{KL}(\cdot||\cdot) \geq 0 $, and Jensen's inequality on $\mathbb{E}[\log|Q(X)|]$ we get \eqref{eq:quant}. 
\end{proof}
\vspace*{-0.1cm}

While Proposition \ref{prop:quant} is useful in bounding the performance  for a given partition $\{u_k\}_{k=1}^M$, an important question in the design of a BA strategy is how the $b$ scanning beams $\{\Phi_i\}_{i=1}^b$ can be optimized to result in the best partition in \eqref{eq:optimization2}. We investigate this question for two scenarios. We first consider the general case where there are no constraints on the scanning beams $\{\Phi_i\}_{i=1}^b$ and next, for practical considerations, we study the scanning beams that are constrained to be contiguous, i.e. the case where the ACR of the scanning beams are contiguous angular intervals.
\vspace*{-0.1cm}
\subsection{No Constraints on the Scanning Beams}
\label{subsec:noncontig}
Theorem \ref{th:sc1} bounds $\overline{U}^*$  and provides a method for designing the scanning beams $\{\Phi_i\}_{i=1}^b$ when there are no constraints on the beams.
\vspace*{-0.1cm}
\begin{Theorem}
\label{th:sc1}
The optimal value of the objective function in optimization problem \eqref{eq:optimization} $\overline{U}^*$ is bounded as 
\begin{align}
\label{eq:ncon}
\frac{2^{h(\Psi)}}{2^{b}} \leq \overline{U}^* \leq \frac{2\pi}{2^b},
\end{align}
where $h(\Psi)$ is the differential entropy of random variable $\Psi$ with PDF in \eqref{eq:equivalence}.
\end{Theorem}
\begin{proof}

\textit{Lower-bound}: A set of BA scanning beams $\{\Phi_i\}_{i=1}^b$ results in a partition $\{u_k\}_{k=1}^M$ of the AoD with $M \leq 2^b$. The BA procedure can be thought of as quantization of the angle $\Psi$ whose distribution is given in \eqref{eq:equivalence}. Using Lemma \ref{lem:eq}, we have the lower-bound in the theorem. 
\begin{figure}[t]
\centering
\includegraphics[width=\linewidth, draft=false]{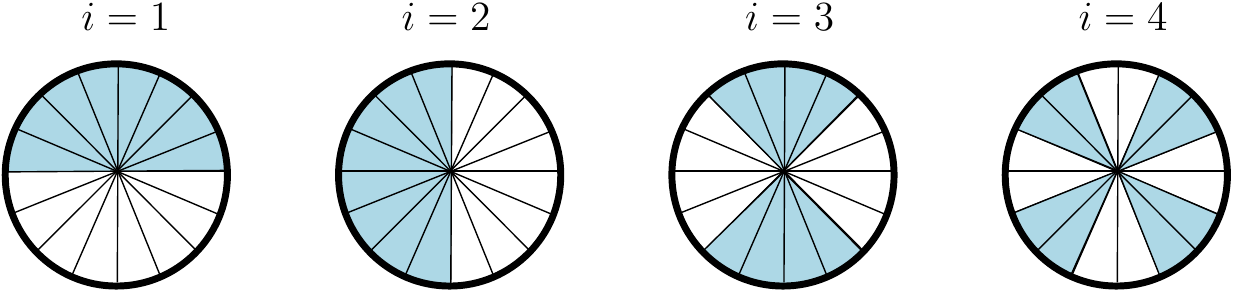}
\caption{An example of the proposed BA strategy for $b = 4$ in the proof of Theorem \ref{th:sc1}.}
\label{fig:exp1}
\vspace*{-0.6cm}
\end{figure}

\textit{Upper-bound}: Consider the following beams
\begin{gather}
\label{eq:beam}
\Phi_1 = \left (0, \pi \right],\quad \Phi_2 = \left(\frac{\pi}{2}, \frac{3\pi}{2} \right] \\ 
\Phi_i = \left\{\left(\frac{\pi}{2^{i-1}}, \frac{3\pi}{2^{i-1}}\right] + \frac{j\pi}{2^{i-3}} \bigg |~j=\left\{0,1,\cdots, 2^{i-2}-1\right\} \right\}, \notag
\end{gather}
for $i\in[b]$. These collection of beams, illustrated in Fig. \ref{fig:exp1} for $b=4$, result in $\overline{U} = \frac{2\pi}{2^{b}}$. 
\end{proof}
\vspace*{-0.2cm}
When the priors on the users' AoDs are not available, one may assume uniform distribution in which case the upper and lower-bounds in Theorem \ref{th:sc1} meet at $\frac{2\pi}{2^b}$, which is achievable by the beams given in \eqref{eq:beam}.

In the following, we first show an important property of the optimal partition $\{u_k^*\}_{k=1}^M$ in \eqref{eq:optimization2}, then using this property, we design the optimal partition and provide a possible set of scanning beams based on the optimal partition.

For a given partition $\{u_k\}_{k=1}^M$, let $a_k = |u_k|$ be the width of the $k^{th}$ region of the partition, hence $ \sum_{k=1}^M a_k =  2\pi$. Without loss of generality, we assume $a_1\leq a_2\leq \cdots \leq a_M$. Consider differential arcs of width $d\theta$ and probability $f_{\Psi}(\theta) d\theta$, where $f_{\Psi}(\cdot)$ is given by \eqref{eq:equivalence}. Clearly, for any two arcs, say at $\theta_i$ and $\theta_j$ with probabilities $f_{\Psi}(\theta_i) d\theta_i \geq f_{\Psi}(\theta_j) d\theta_j$ the one with larger probability (in this case $\theta_i$) should belong to the region of the partition with smaller width in order to minimize the objective value in \eqref{eq:optimization2}, otherwise the association can be reversed to achieve smaller value in \eqref{eq:optimization2}. Therefore, we have the following proposition.

\vspace*{-0.1cm}
\begin{Proposition}
\label{prop:creg_ncon} For the optimal partition $\{u_k^*\}_{k=1}^M$, any two differential arcs at $\theta_i$ and $\theta_j$ in partition regions $u^*_i$ and $u^*_j$, respectively, satisfy $f_{\Psi}(\theta_i) \geq f_{\Psi}(\theta_j)$ if $|u_i| < |u_j|$.
\end{Proposition}
\vspace*{-0.1cm}

Let us consider the case that $f_{\Psi}(\cdot)$ is a monotonic function over the interval $(0,2\pi]$. Using the Proposition \ref{prop:creg_ncon}, we note that the optimal partition $\{u_k^*\}_{k=1}^{M}$ consists of $M$ contiguous intervals $u^*_i = (x^*_i, x^*_{i+1}]$ that are defined only based on the boundary points $x^*_i, i \in [M]$ with $x^*_{M+1} = x^*_1$. This simplifies the search for of the optimal partition $\{u_k^*\}_{k=1}^{M}$ to finding the boundary values $x^*_i, i \in [M]$ which are the solutions to the following optimization
\begin{align}
\label{eq:OpCon}
\begin{aligned}
\argmin_{x_i, i\in [M]} \quad  &\sum_{i=1}^{M} (x_{i+1} - x_i)\int_{x_i}^{x_{i+1}} f_{\Psi}(x)dx. \\
\end{aligned}
\end{align}


Building on the design idea for monotonic PDF for AoD, we note that there is always a one-to-one mapping $g: (0,2\pi] \rightarrow (0,2\pi]$ where PDF of $g(\Psi)$ is monotonic over the interval $(0,2\pi]$ for any PDF of AoD $\Psi$. Hence, by designing the optimal partition for $g(\Psi)$, one can find the corresponding optimal partition for $\Psi$.

Finally, we need to define the optimal scanning beams $\{\Phi_i\}_{i=1}^b$ in terms of optimal partition $\{u_k^*\}_{k=1}^{M}$. Combining, we have
\vspace*{-0.1cm}
\begin{Theorem}
\label{th:OpBeam}
The optimal scanning beams in \eqref{eq:optimization} are given by 
\begin{align}
\label{eq:nconb}
&\Phi_i^* = \cup_{j \in S_i}  u^*_j, i\in [b],
\end{align}
where $u_i^*=g^{-1}((x_i^*, x_{i+1}^*])$, with $x_i, i\in[2^b]$ being solutions of \eqref{eq:OpCon}, $g(\cdot)$ is a mapping such that $f_{g(\Psi)}(\cdot)$ is monotone over the interval $(0, 2\pi]$ , and $S_i$ is the set of all values in $\{1, \ldots, 2^b\}$ for which the $i^{th}$ bit in their binary representation is zero.
\end{Theorem}
\vspace*{-0.1cm}
The optimization in \eqref{eq:OpCon} can be solved using standard linear programming software such as CVX \cite{cvx} when the objective function is convex. While, in general, this optimization can be non-convex for arbitrary PDFs, for small values of $b$, one can find a near optimal solution by an exhaustive search over a fine grid on $(0,2\pi]$. Trivially, for large values of $b$, the upper and lower-bounds in Theorem \ref{th:sc1} get closer and the explicit construction for the upper-bound in Theorem \ref{th:sc1} performs close to the optimal.

In practice, it is not often feasible to realize the beams with non-contiguous ACR and sharp roll-offs. Therefore, 
while our approach provides the characteristics of the optimal scanning beams described in Theorem \ref{th:OpBeam}, it may not necessarily lead to a practical NI-BA. To address this issue, we consider the problem of NI-BA optimization with contiguous beams next. 

\vspace*{-0.2cm}
\subsection{Scanning with Contiguous Beams}
In this section, we constrain the ACR of the scanning beams to be contiguous intervals. When using contiguous scanning beams, the resulting URs may be non-contiguous. 
In the following, we show that for the optimal contiguous scanning beams $\{\Phi_i\}_{i=1}^b$ the resulting URs $\{u_k\}_{k=1}^M$ are also contiguous. The importance of this result is that it not only addresses the practicality of the beams  used in the BA phase, but also in the data communication phase where the BS uses  a beam covering user's UR to transmit to that user.

\begin{Proposition}
\label{prop:creg_con}
For $b$ contiguous scanning beams $\{\Phi_i\}_{i=1}^b$ there are at most $M = 2b$ URs in the resulting partition $\{u_k\}_{k=1}^M$. Furthermore, without loss of generality, the URs in the optimal partition minimizing \eqref{eq:optimization2} are contiguous.
\end{Proposition}
\begin{proof}
A contiguous scanning beam results in an interval, which is described by its two end points. The collection of $b$ contiguous beams are thus represented by at most $2b$ such end points, resulting in at most $M=2b$ URs. Given contiguous scanning beams $\{\Phi_i\}_{i=1}^b$ with partition $\{u_k\}_{k=1}^M, M < 2b$. we can create the set $\{c_i\}_{i=1}^{2b}$ including all the end points of the scanning beams in an ascending manner i.e., $c_i \leq c_j$ for $i\leq j$. Define the beams $\{\Phi_i'\}_{i=1}^b$ as $\Phi_i' = (c_i, c_{i+b}]$. These beams form the partition $\{u_k'\}_{k=1}^{M'}, M' \geq M$. It is easy to verify that these URs are contiguous and the partition $\{u_k\}_{k=1}^M$ is formed by grouping $u_k'$s together. Therefore, $\{u_k'\}_{k=1}^{M'}$ has less expected width then $\{u_k\}_{k=1}^{M}$.
\end{proof}

Theorem \ref{th:sc1c} below bounds the optimal performance when contiguous scanning beams are used.

\begin{Theorem}
\label{th:sc1c}
The optimal value of the objective function $\overline{U}^*$ in optimization problem \eqref{eq:optimization} when contiguous scanning beams are used is bounded as 
\begin{align}
\label{eq:con}
\frac{2^{h(\Psi)}}{2b} \leq \overline{U}^* \leq \frac{\pi}{b}. 
\end{align}
\end{Theorem}
\begin{proof}
The lower-bound is derived similar to that of Theorem \ref{th:sc1} using the maximum number of regions as $2b$ following Prop. \ref{prop:creg_con}.
\begin{figure}[t]
\centering
\includegraphics[ width=\linewidth, draft=false]{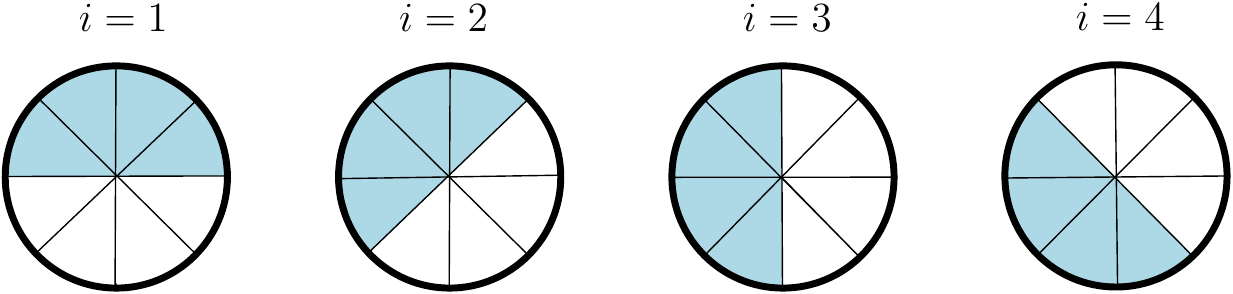}
\caption{An example of the proposed BA scheme when scanning beams are contiguous for $b = 4$ in the proof of Theorem \ref{th:sc1c}.}
\label{fig:exp2}
\vspace*{-0.6cm}
\end{figure}
For the upper-bound, consider the following beams 
\begin{align}
\Phi_i = \left(\frac{(i-1)\pi}{b}, \pi + \frac{(i-1)\pi}{b}\right], i\in [b].
\end{align}
These beams, illustrated in Fig. \ref{fig:exp2} for $b=4$, result in expected uncertainty of $\frac{\pi}{b}$.
 \end{proof}
 
Using Proposition \ref{prop:creg_con} and following similar approach as in Theorem \ref{th:OpBeam}, the following theorem provides the optimal set of scanning beams for optimization problem in \eqref{eq:optimization}. Note that for the uniform distribution on the AoD, the optimal set of contiguous beams corresponds to the ones used in the upper-bound of Theorem \ref{th:sc1c} and illustrated in Fig. \ref{fig:exp2} for $b=4$.

\begin{Theorem}
\label{th:OpBeamc}
The optimal contiguous scanning beams in \eqref{eq:optimization} are given by 
\vspace{-0.2cm}
\begin{align}
\label{eq:conb}
\Phi_i = (x^*_i, x^*_{i+b}], i\in [b],
\end{align}
where, $x_i^*, i\in[2b]$ are the solution to the optimization \eqref{eq:OpCon} with $M = 2b$ regions.
\end{Theorem}

We conclude this section by comparing the performance of our approach with an adaptation of ES to our model. The scanning beams in ES are defined by $\Phi_i = ((i-1)\frac{2\pi}{b+1}, i\frac{2\pi}{b+1}], i\in[b]$. To have a fair comparison, we assume that we use the optimal contiguous scanning beams from Theorem \ref{th:OpBeamc} since ES beams are also contiguous. Considering a uniform PDF on the AoD, it can be seen that the performance of ES is $\overline{U}_{ES}=\frac{2\pi}{b+1}$ while the performance of our approach is $\overline{U}^*=\frac{\pi}{b}$ resulting in a factor of two gain for sufficiently large values of $b$. Furthermore, the gain can be larger than $2$ since unlike
ES, the proposed method exploits the distribution of AoD.
For instance, Fig. \ref{fig:exp3} shows that the gains obtained by using optimized contiguous beams over ES are $1.48\times,1.87\times,2.17\times,2.22\times,2.31\times$ for $b=2,3,4,5,6$, respectively. This figure considers a two-user scenario where both users have piecewise uniform distributions and $w_1=w_2=1/2$. AoD of the first and second user are uniformly distributed over the first and third quadrant, respectively, with probability $0.9$ and uniformly distributed over the remaining quadrants with probability $0.1$. The figure also indicates that the difference between upper-bound and lower-bound presented in Theorem \ref{th:sc1c} is small for relatively large values of $b$. Therefore, designing the beams using uniform prior does not degrade the performance much for large values of $b$. Fig. \ref{fig:exp3b} illustrates the optimal contiguous scanning beams along with the corresponding URs for $b=4$. We observe that the optimal strategy creates narrower URs in the first and third quadrants where the PDFs $f_{\Psi_1}$ and $f_{\Psi_2}$ are more concentrated.
\begin{figure}[t]
\centering
\includegraphics[ width=0.8\linewidth, draft=false]{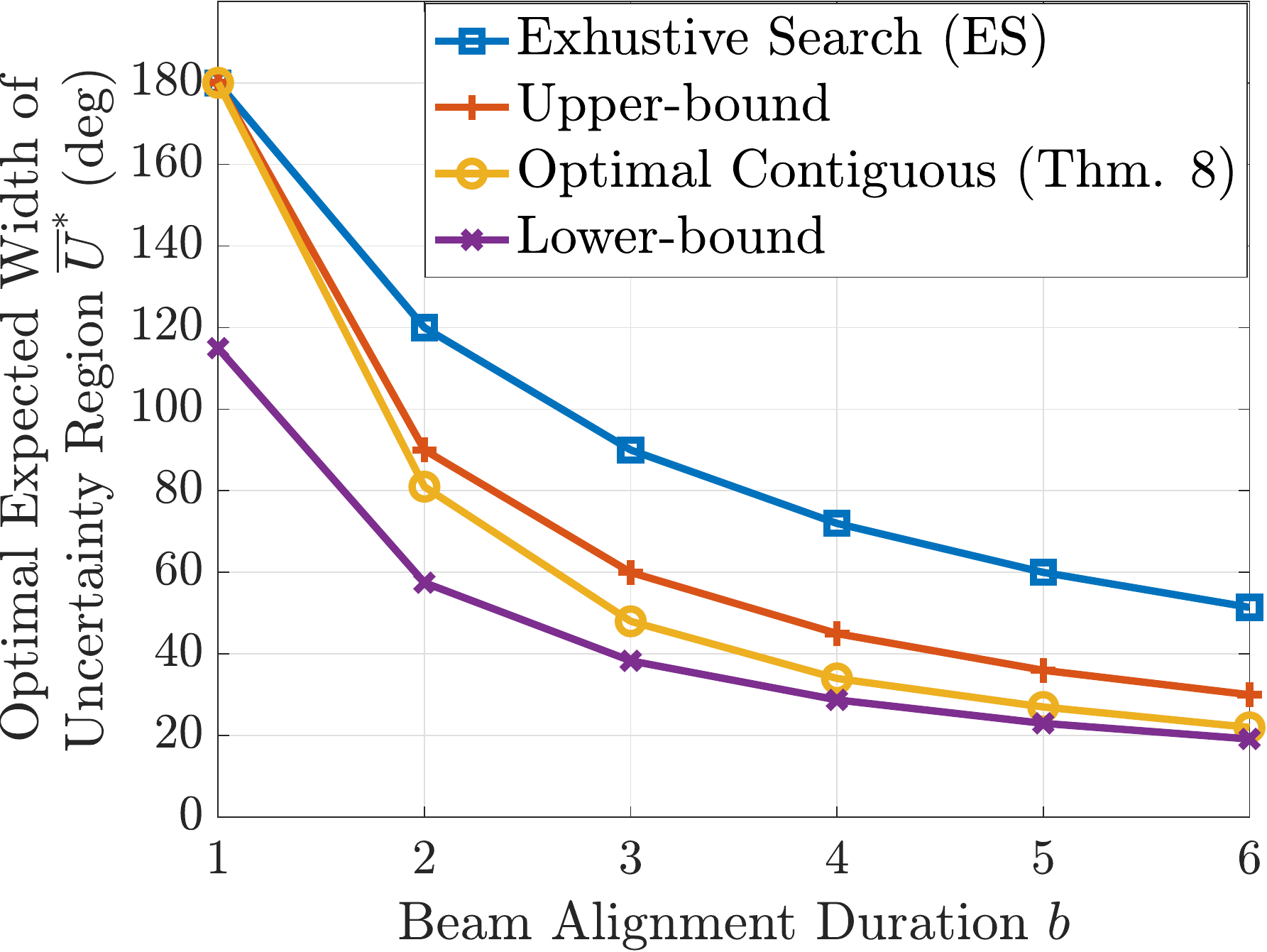}
\caption{Performance of the optimal contiguous scanning beams (Theorem \ref{th:OpBeamc}), upper-bound and lower-bound on the optimal objective (Theorem \ref{th:sc1c}), and ES.}
\label{fig:exp3}
\vspace*{-0.3cm}
\end{figure}

\begin{figure}[t]
\centering
\includegraphics[ width=\linewidth, draft=false]{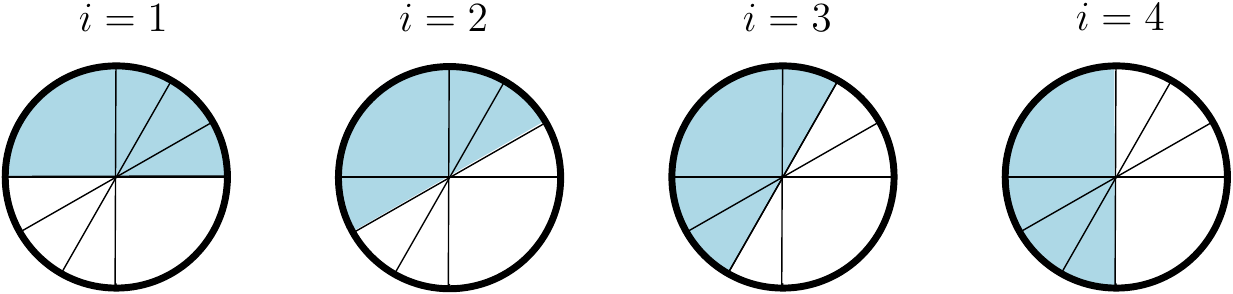}
\caption{Scanning beams and corresponding URs of the example in Fig. \ref{fig:exp3} when $b =4$.}
\label{fig:exp3b}
\vspace*{-0.6cm}
\end{figure}

\vspace*{-0.2cm}
\section{Conclusion}
In this paper, we have investigated the multi-user non-interactive beam alignment problem in mmWave systems where the objective is to minimize the average of expected width of remaining uncertainty regions on the AoDs after receiving users' feedbacks. The problem is investigated for two scenarios: i) when there are no constraints on the beam alignment scanning beams and ii) when they have to be contiguous. We have provided information theoretic bounds on the optimal performance and proposed methods to design optimal scanning beams for both scenarios. 
Additionally, we have proved that using optimal contiguous scanning beams will lead to contiguous uncertainty regions suggesting that a contiguous beam can also be used for data communications after the beam alignment. Furthermore, we have shown that the performance of our proposed approach improves that of exhaustive search, a well-known approach considered in 5G standard, for reasonable values of $b$, by at least a factor of two.

\balance
\bibliographystyle{IEEEtran}
\bibliography{ref}

\end{document}